\documentclass[journal,onecolumn,11pt,final]{IEEEtran}
\usepackage[T1]{fontenc}
\usepackage[latin9]{inputenc}
\usepackage{float}
\usepackage{amsthm}
\usepackage{amsmath}
\usepackage{amssymb}
\usepackage{dsfont}
\usepackage{graphicx}
\usepackage{esint,bbm,enumerate,color}
\usepackage[unicode=true,
 bookmarks=true,bookmarksnumbered=true,bookmarksopen=true,bookmarksopenlevel=1,
 breaklinks=false,pdfborder={0 0 0},backref=false,colorlinks=false]
 {hyperref}
\hypersetup{pdftitle={Your Title},
 pdfauthor={Your Name},
 pdfpagelayout=OneColumn,pdfnewwindow=true,pdfstartview=XYZ,plainpages=false}

\newcommand{\sd}{\Sigma\Delta}

\newcommand{\R}{\mathbb{R}}

\newcommand{\A}{\mathcal{A}}
\makeatletter

\DeclareMathOperator{\diag}{diag}
\DeclareMathOperator{\supp}{supp}
\DeclareMathOperator{\rank}{rank}

\theoremstyle{plain}
\newtheorem{thm}{\protect\theoremname}

\newtheorem{lem}[thm]{\protect\lemname}
\theoremstyle{definition}
\newtheorem{defn}[thm]{\protect\definitionname}

\newtheorem{remark}{Remark}
\theoremstyle{plain}
\theoremstyle{plain}
\newtheorem{cor}[thm]{\protect\corollaryname}

\theoremstyle{definition}

\providecommand{\corollaryname}{Corollary}
\providecommand{\propname}{Proposition}
\providecommand{\definitionname}{Definition}
\providecommand{\theoremname}{Theorem}
\providecommand{\lemname}{Lemma}
\providecommand{\assumptionname}{Assumption}

\makeatother

\providecommand{\corollaryname}{Corollary}
\providecommand{\definitionname}{Definition}
\providecommand{\theoremname}{Theorem}

\newcommand{\M}{\mathcal M}

\renewcommand{\P}{\mathbb P}

\newcommand{\SD}{\Sigma\Delta}

\providecommand{\ep}{\epsilon}

\begin{document}

\title{Quantization for Low-Rank Matrix Recovery}
\author{Eric Lybrand, Rayan Saab
 \thanks{Eric Lybrand and Rayan~Saab are with the Mathematics Department at the University
    of California, San Diego.%
  }}

\maketitle
\begin{abstract}
We study Sigma-Delta \((\SD)\) quantization methods coupled with appropriate reconstruction algorithms for digitizing randomly sampled low-rank matrices. We show that the reconstruction error associated with our methods decays polynomially with the oversampling factor, and we leverage our results to obtain root-exponential accuracy by optimizing over the choice of quantization scheme.  Additionally, we show that a random encoding scheme, applied to the quantized measurements, yields a near-optimal exponential bit-rate. As an added benefit, our schemes are robust both to noise and to deviations from the low-rank assumption. In short, we provide a full generalization of analogous results, obtained in the classical setup of bandlimited function acquisition, and more recently, in the finite frame and compressed sensing setups to the case of low-rank matrices sampled with sub-Gaussian linear operators. Finally, we believe our techniques for generalizing results from the compressed sensing setup to the analogous low-rank matrix setup is applicable to other quantization schemes. 

\end{abstract}
\begin{IEEEkeywords}
Compressed sensing, quantization, exponential accuracy, rate-distortion, low-rank, one-bit
\end{IEEEkeywords}

\section{Introduction}

Let $\M: \R^{n_1\times n_2} \to\R^m $ be a linear map that acts on matrices $X$ to produce measurements \begin{equation}y= \M(X) = \sum\limits_{i=1}^m \langle X, A_i \rangle e_i\label{eq:lin_op},\end{equation} where the vectors $e_i$ are the standard basis vectors for $\R^m$, and each $A_i$ is a matrix in $\R^{n_1\times n_2}$, $i\in \{1,...,m\}.$ Here, the inner product is the standard Hilbert-Schmidt inner product given by $\langle Y,Z \rangle = \sum_{i,j} Y_{ij}Z_{ij}$. 
Note that for every linear operator $\M$ as above, there exists an $m \times (n_1 n_2)$ matrix $A_\M$, such that for all $X \in \R^{n_1\times n_2}$ $$\M(X) = A_\M \vec{X}.$$
Here $\vec{X} \in \R^{n_1n_2}$, the {vectorized} version of the matrix $X$, is obtained by stacking the columns of $X$. Low-rank matrix recovery is concerned with approximating a rank $k$ matrix $X$ from $y$, knowing the operator $\M$. It is primarily interesting in the regime where $m \ll n_1 n_2$, and many recent results propose recovery algorithms and prove recovery guarantees when $m \geq C k\max\{n_1,n_2\}$ where \(C > 0\) is some absolute constant \cite{candes2011tight, CR09, KRM10, RFP07}. For example, when the entries of the matrix representation $A_\M$ are independent Gaussian or sub-Gaussian random variables, and $y = \M(X) + e$ with $\|e\|_2 \leq \varepsilon$, one can solve the convex optimization problem
\begin{equation}
{X^\sharp}:=\arg\min_Z \| Z \|_* \quad \text{subject to } \quad \| \M(Z) - y \|_2 \leq \varepsilon. 
\end{equation}
Then, with high probability on the draw of $\M$, and uniformly for all $n_1\times n_2$ matrices $X$, we have 
\begin{equation}
\| X^{\sharp} - X\|_F \leq C\left(\frac{\sigma_k(X)_*}{\sqrt{k}} + \varepsilon \right),
\end{equation}
as shown in, e.g., \cite{candes2011tight}. Above, \( \| \cdot \|_F \) denotes the Frobenius norm \(\|X\|_F = \sqrt{\sum_{i,j}A_{i,j}^2}\) on matrices induced by the Hilbert-Schmidt inner product, $\| Z\|_*$ denotes the nuclear norm of $Z$, i.e., the sum of its singular values, and \[\sigma_k(Z)_*:= \min\limits_{\rank(V)=k} \| Z - V\|_*\]
denotes the error,  measured in the nuclear norm, associated with the best rank $k$ approximation of a matrix.%

\subsection{Background and prior work}
Low-rank matrix recovery has seen a wide range of applications, ranging from quantum state tomography \cite{GLFB10} and collaborative filtering \cite{HPlabs08} to  sensor localization \cite{sensor10} and face recognition \cite{candesPCA}, to name a few.

Nuclear norm minimization was  proposed by Fazel in \cite{fazel2002} as a means of finding a matrix of minimial rank in a given convex set. Fazel motivates this through the observation that nuclear norm minimization is the convex relaxation of the rank minimization problem, which has been shown to be NP-hard.  Since then, and with the advent of compressed sensing, there has been much work on recovering low-rank matrices from linear measurements. For example, \cite{RFP07} considers recovering low-rank matrices given random linear measurements, and establishes recovery guarantees given that the sampling scheme satisfies the matrix restricted isometry property, which we define in the subsequent section.  Perhaps not surprisingly, this analysis closely follows that of sparse vector recovery under \( \ell_1\) minimization, as it is known that random ensembles of linear maps satisfy the matrix restricted isometry property with high probability \cite{FR13, RFP07}.  This led to a flurry of papers on nuclear norm minimization for matrix recovery in various contexts, see \cite{candesPCA, lin10, CT10} for example. 

While the theoretical results on nuclear norm minimization have been promising, convex optimization practically necessitates the use of digital computers for recovering the underlying matrix. It behooves the theory, therefore, to take into account that the measurements must be converted to bits so that numerical solvers can handle them. Indeed, quantization is the necessary step in data acquisition by which measurements taking values in the continuum are mapped to discrete sets. Without any claim to comprehensiveness, we are aware of the following developments on quantization in the low-rank matrix completion setting, i.e., the setting where one quantizes a random subset of the entries of the matrix directly. 

Davenport and coauthors in \cite{davenport14} consider recovering a rank $k$ matrix $X\in \R^{n_1\times n_2}$ given \(1\)-bit measurements of a subset of the entries, sampled according to a distribution that \emph{may depend on the entries}. They recover an estimate $\hat{X}$ of the sampled matrix through maximum likelihood estimation with a nuclear norm constraint and derive error bounds which decay, as a function of the number of measurements $m$, like \( O(m^{-1/2}) \).

Shortly thereafter Cai and Zhou in \cite{CZ13} consider reconstruction given 1-bit measurements of the entries under more general sampling schemes of the indices. Unlike the argument in \cite{davenport14}, Cai and Zhou impose a max-norm constraint on the maximum likelihood estimation to enforce the low-rank condition. Under this regime the scaled Frobenius norm error decay is also \( O(m^{-1/2}) \).

Bhaskar and Javanmard in \cite{BJ15} modify the optimization problem of \cite{davenport14} so that it now imposes an exact rank constraint in place of the nuclear norm. This yields a non-convex problem with associated computational challenges. Nevertheless, assuming one can solve this hard optimization problem, they obtain an error estimate that decays like \( O(m^{-4}) \), at the added cost of a much increased constant that scales like $n^7k^3$.
 

Proceeding towards more general quantization alphabets, \cite{LSB14} consider low-rank matrix completion via nuclear norm penalized maximum liklihood estimation given quantized measurements of the entries with unknown quantization bin boundaries. They propose an optimization procedure which learns the quantization bin boundaries and recovers the matrix in an alternating fashion. No theoretical guarantees are given to delineate the relationship between the number of measurements and the reconstruction error.

Authors in \cite{LKMS14} propose a low-rank matrix recovery algorithm given quantized measurements of the entries from a finite alphabet under some sampling distribution of the indices. As the aforementioned schemes have done, they propose a maximum liklihood estimation but with a nuclear norm constraint to enforce the low-rank condition.  Specifically, given \( m \ge C \max\{n_1, n_2\} \log(\max\{n_1, n_2\}) \) measurements where \(C > 0\) is a universal constant, they show that the scaled Frobenius error decays like 
\(m^{-1}\).

In contrast to the above works, we study the quantization problem in the low-rank matrix recovery setting given linear measurements of the form \eqref{eq:lin_op}, where the matrices $A_i$ are sub-Gaussian. 

\subsection{Contributions}
To the best of our knowledge, we provide the first theoretical guarantees of low-rank matrix recovery from \(\SD\) quantized sub-Gaussian linear measurements. Our result holds for stable \(\SD\) quantizers (defined in Section \ref{sec:SDprelims}) of arbitrary order and our bounds apply to the particular case of $1$-bit quantization; that is, we can recover scaling information in this setting.
Thus, we generalize a result from \cite{SWY15} that recovers sparse vectors from quantized  noisy measurements so that it now applies to the low-rank matrix setting, as shown in Theorem \ref{thm: main}. Our main tool for achieving this extension is a modification of the technique of Oymak et al. \cite{oymak2011simplified} for converting compressed sensing results to the low-rank matrix setting. We show that the reconstruction error under constrained nuclear norm minimization is bounded by
\begin{align*}
\|{X^\sharp}-X\|_F \leq C\left(\left(\frac{m}{\ell}\right)^{-r+1/2}\beta+\frac{\sigma_k(X)_*}{\sqrt k}+\sqrt{\frac{m}{\ell}}\epsilon \right)
\end{align*}
thus showing that our reconstruction scheme is robust to noise and to the low-rank assumption. Above, \( r \) denotes the order of the \( \SD \) scheme and \(\beta \) the step-size of the associated alphabet (see Section \ref{sec:SDprelims}), \( \ell \) is of order \( k \max\{n_1,n_2\} \), and \( \frac{m}{\ell} \)  denotes the oversampling factor.  Note that in the case of rank $k$ matrices, with no measurement noise, our reconstruction error decays polynomially fast, namely as $m^{-r}$, thereby greatly improving on the rates obtained in the works cited above. Furthermore, by optimizing over the order of the \( \SD \) reconstruction scheme, we show in Corollary \ref{cor: root exp} that our procedure attains root-exponential accuracy with respect to the oversampling factor. This generalizes the error decay seen in \cite{SWY15} for vectors.
   
The robustness of the main result extends beyond quantization. We show in Corollary \ref{cor: encoding} that we can further reduce the total number of bits, by encoding the quantized measurements using a discrete Johnson-Lindenstrauss \cite{JL84} embedding into a lower dimensional space. 
The resulting dramatic reduction in bit-rate is coupled with only a small increase in reconstruction error. This, in turn yields an exponentially decaying, i.e., optimal, relationship between number of bits and reconstruction error. 

Finally, we remark that the techniques used herein can be used to derive analogous results for other quantization schemes that share certain properties of \( \SD \) quantization. Namely, suppose one is given a quantization map \( \mathcal{Q} \) and a bijective linear map \( T: \R^n \to \R^n \) which satisfy \( \| T(y - \mathcal{Q}(y)) \| < C \) for some norm \( \| \cdot \| \) and some constant \( C \) that may depend on the quantization technique but not on the dimensions. Then, the proof of Theorem \ref{thm: main}, with a suitably altered decoder, can likely be modified to produce an analogous result for the new quantization scheme.

\section{Preliminaries}
\subsection{Notation}
For \(x \in \R^n\), let \(\supp(x)\) denote the set of indices \(i\) for which \(x_i\) is non-zero, and $\Sigma_k^n:=\{x \in \R^n, |\supp(x)| \leq k\}$  be the set of all $k$-sparse vectors in $\R^n$. For a matrix $A\in \R^{n_1\times n_2}$, we will denote its singular values by $\sigma_i(A)$ for $i=1,..., n$ where $n:=\min\{n_1,n_2\}$ and $\sigma_1(A)\geq\sigma_2(A)\geq ...\geq \sigma_n(A)$. 
We will require the  definitions of the well known restricted isometry property (RIP), both for linear operators acting on sparse vectors and for linear operators acting on low-rank matrices. 

\begin{defn}[\emph{vector-RIP} (e.g., \cite{CRT05})] We say a linear operator $\Phi: \R^{n} \to \R^{m}$ satisfies the vector-RIP of order $k$ and constant $\delta_k$, if for all $x \in \Sigma_k^n$,
\[ (1-\delta_k)\|x\|_2^2 \leq \|\Phi x\|_2^2 \leq (1+\delta_k)\|x\|_2^2. \]
\end{defn}
\begin{defn}[\emph{matrix-RIP}] We say a linear operator $\M: \R^{n_1\times n_2} \to \R^{m}$ satisfies the matrix-RIP of order $k$ and constant $\delta_k$, if for all matrices $X$ of rank $k$ or less we have,
\[ (1-\delta_k)\|X\|_F^2 \leq \|\M( X) \|_2^2 \leq (1+\delta_k)\|X\|_F^2. \]
\end{defn}
\begin{defn}[Restriction \cite{oymak2011simplified}] Let $\M: \R^{n_1\times n_2}  \to \R^m $ be a linear operator and assume without loss of generality that $n_1 \leq n_2$. Given a pair of matrices $U$ and $V$ with orthonormal columns, define $\M_{U,V}$, the $(U,V)$ restriction of $\M$ by \footnote{Here, given a vector $x\in \R^n$, $\diag(x)=X$  is a diagonal matrix in $ \R^{n\times n}$ with $X_{ii} = x_i$ for $i\in \{1,..., n\}$.  }
\begin{align}
 \M_{U,V}: \R^{n_1} &\to \R^m  \nonumber \\
 x &\mapsto \M(U \diag(x) V^*). \nonumber
\end{align}
\end{defn}

\subsection{Preliminaries on $\SD$ quantization}\label{sec:SDprelims}
$\SD$ quantizers were first proposed in the context of digitizing oversampled band-limited functions by \cite{inose1963unity}, and their mathematical properties have been studied since. In this band-limited context, the $\SD$ quantizer takes in a sequence of point evaluations  of the function sampled at a rate exceeding the critical Nyquist rate and produces a sequence of \emph{quantized} elements, i.e., elements from a finite set.  So, the $\sd$ quantizer is associated with this finite set, say $\A\subset \R$ (called the quantization alphabet), and also with a \emph{scalar quantizer}
\begin{align}
Q_\A: \R &\to \A \nonumber \\
z &\mapsto \arg\min_{v\in\A} | v - z |.
\end{align}
$\SD$ schemes build on scalar quantization by incorporating a state  variable sequence $u$, which is recursively updated. In an $r$th order $\SD$ scheme, a function, say $\rho_r$, of  $r$ previous values of $u$ and the current measurement  are fed into the scalar quantizer to produce an element from $\A$. For example, in the band-limited context the measurements are simply the pointwise evaluations of the function. Defining $u_i = 0$ for $i\leq 0$, and denoting the measurements by $y_i$ we have the recursion: 
\begin{align}\label{eq:state}
q_i &= Q_\A ( \rho_r (u_{i-1},...u_{i-r+1},y_i) )\\
(\Delta^r u)_i  &=  y_i - q_i.
\end{align}
Here $(\Delta u)_i := u_i - u_{i-1}$, and $\Delta^r u := \Delta (\Delta^{r-1}u)$.
Thus, the $r$th order $\sd$ quantizer updates the state variables as a solution to an $r$th order difference equation.  To give a concrete example, 
   the simplest $1$st order $\SD$ scheme operates by running the following recursion: 
\begin{align}
q_i &= Q_\A ( y_i + u_{i-1} )\\
u_i &=  u_{i-1} + y_i - q_i.
\end{align}
Usually, the alphabet $\A$ associated with $\sd$ quantizers is of the form \[\A:=\{ \pm (j-1/2)\beta, j = 1, ... ,L\}.\]  We refer to such an $\A$ as a $2L$-level alphabet with step-size $\beta$. In particular, when $L=1$, we have a $1$-bit alphabet. 

For reasons related to building a circuit that implements the $\sd$ quantization scheme and bounding the reconstruction error, an important consideration is the so-called stability of the $\SD$ scheme. A stable $r$th order $\sd$ scheme produces \emph{bounded} state variables with \begin{equation}\|u\|_\infty < \gamma(r), \label{eq:bound_state}\end{equation} whenever $\|y\|_{\infty}$ is bounded above. Above, \(\gamma(r)\) is some constant which may depend on \(r\).
For example, for the \(2L\)-level alphabet described above, coupled with a particular choice of $\rho_r$ and \(\SD\) order \( r\) it is sufficient to choose \( L \ge 2 \lceil\frac{\|y\|_{\infty}}{\beta}\rceil + 2^r + 1 \) to guarantee \eqref{eq:bound_state} holds with \( \gamma(r) = \beta/2\) \cite{BYP04, daub-dev}. Note that with such a choice the size of the alphabet grows exponentially as a function of the $\SD$ order. On the other hand, given a fixed alphabet,  \cite{daub-dev} constructed the first family of functions $\rho_r$ with associated stability constants  $\gamma(r)$. Subsequently, the dependence on $r$ was improved upon by \cite{G-exp} and \cite{DGK10} via different constructions of $\rho_r$.   In these papers it was shown that $\sd$ quantized measurements of a band-limited function $f$, sampled at a rate $\lambda$ times the critical Nyquist rate, can be used to obtain an approximation $\hat{f}$ of $f$ satisfying
$$\|\hat{f}-f \|_\infty \leq C \gamma(r)\lambda^{-r}.$$
By optimizing the right hand side above, i.e., $\gamma(r)\lambda^{-r}$ as a function of $r$, \cite{G-exp} and \cite{DGK10} obtain the error rates $$\|\hat{f}-f\|_\infty \leq Ce^{-c\lambda}$$ where $c<1$ is a known constant depending on the family of schemes. 

Outside of the band-limited context, $\sd$ schemes were proposed and studied for quantizing finite-frame coefficients \cite{IS13, BYP04, KSW12, KSY14} as well as compressed sensing coefficients \cite{G10, FKS17, BJKS14, SWY15}. In both these contexts, given a linear map $\Phi: \R^{n} \mapsto \R^m$, absent noise, one obtains measurements \[ y = \Phi x \]  of a vector $x \in \mathcal{X} \subset \R^n$ and quantizes using an $r$th order stable $\sd$ scheme. To ensure boundedness of the resulting state variable, typically one has $\mathcal{X} \subset \{x\in \R^n: \| \Phi x\|_\infty < 1 \}$. One may also enforce additional restrictions on elements of $\mathcal{X}$, such as $k$-sparsity.   Here, as before, one runs a stable $r$th order $\sd$ quantization scheme 
\begin{align}
Q_{\SD}^{(r)}: \R^m \to  \mathcal{A}^m.
\end{align}
Writing the $\SD$ state equations \eqref{eq:state} in matrix-vector form yields
\begin{equation}\label{eq:state_discrete}
y - q = D^r u 
\end{equation}
where $D\in \R^{m\times m}$ is the lower bi-diagonal difference matrix with $1$ on the main diagonal and $-1$ on the sub-diagonal. 
In analogy with the band-limited case, here one defines the oversampling factor as the ratio of the number of measurements $m$ to the minimal number $m_0$ needed to ensure that $\Phi$ is injective (or stably invertible) on $\mathcal X$.  For example $\lambda:= \frac{m}{n}$ in the finite-frames setting when $\mathcal{X}$ is the Euclidean ball, and  $\lambda: = \frac{m}{k \log{n/k}} $ in the compressed sensing context when $\mathcal{X}$ is the intersection of the Euclidean ball with the set of $k$-sparse vectors in $\R^n$.  As in the band-limited context, one wishes to bound the reconstruction error as a function of $\lambda$. A typical result states that provided $\Phi$ satisfies certain assumptions, there exists a reconstruction map 
\begin{align}
\mathcal{D}: \A^m \to \R^n 
\end{align}
such that for all $x\in \mathcal{X}$ and $\hat{x}:= \mathcal{D} (Q_{\sd}^{(r)} (\Phi x)) $, 
\[ \|x - \hat{x} \|_2 \leq C \lambda^{-\alpha (r-1/2)} \]
where $\alpha\leq 1$ is a parameter that,  in the case of random measurements, controls the probability with which the result holds.
Most relevant to this work \cite{SWY15} proposes recovering \emph{arbitrary}, that is, not necessarily strictly sparse, vectors in \( \R^n \) from their \emph{noisy} $\SD$-quantized compressed sensing measurements by solving a convex optimization problem. In particular, one obtains the approximation $\hat{x}$ from $q:=Q_{\SD}^{(r)}(\Phi x + e)$, where $\|e\|_\infty \leq \varepsilon$ via
\begin{align}\label{eq:decoder_intro}
(\hat{x},\hat{\nu}) :=  \arg\min\limits_{(z,\nu)}\|z\|_1 \  \text{ subject to } &  \|D^{-r}(\Phi z+\nu-q )\|_2 \leq \gamma(r)\sqrt{m}  \notag \\
\text{\ \ and\ \ } & \|\nu\|_2\leq \epsilon \sqrt m.
\end{align} 
 Then, \cite{SWY15} shows that the reconstruction error due to quantization decays polynomially in the number of measurements, while maintaining stability and robustness against noise in the measurements and deviations from sparsity. Specifically, defining $$\sigma_k(x):=\arg\min\limits_{v\in\Sigma_k^n} \|x-v\|_1,$$ the following theorem holds.


\begin{thm}\cite{SWY15}\label{thm:SWY} Let $k,\ell,m,n$ be integers, and let $P_\ell:\R^m \to \R^\ell$ be the projection onto  the first $\ell$ coordinates. Let $D^{-r}=U\Sigma V^*$ be the singular value decomposition of $D^{-r}$ and let $\Phi$ be an $m\times n$ 
matrix such that $\frac{1}{\sqrt{\ell}}P_\ell V^* \Phi$ has the vector-RIP of order $k$ and constant $\delta_k < 1/9$. Then, for all $x\in\R^n$  satisfying $\|\Phi x\|_\infty \leq \mu <1$ and all $e$, $\|e\|_\infty\leq \epsilon<1-\mu$, the solution $\hat{x}$ of \eqref{eq:decoder_intro} with $q= Q_{\sd}^{(r)} (\Phi x + e)$  satisfies
\begin{equation}\label{eq:l2err_intro}
\|\hat{x}-x\|_2 \leq C\left(\big(\frac{m}{\ell}\big)^{-r+1/2}\beta+\frac{\sigma_k(x)}{\sqrt k}+\sqrt{\frac{m}{\ell}}\epsilon \right).
\end{equation}
Above $C$ does not depend on $m, \ell,n$. 
\end{thm}

The proof of Theorem \ref{thm:SWY} reveals that a more general statement is true. Indeed, it turns out that the only assumptions on $\hat{x}$ needed 
are that it satisfies the constraints in \eqref{eq:decoder_intro} and that $\|\hat{x}\|_1 \leq \| x\|_1$. Moreover, the only assumption needed on $q$ is that it satisfies the state variable equations \eqref{eq:state_discrete}, and need not belong to $\mathcal{A}^m$. We will use this generalization in proving our main result, and we state it below for convenience.

\begin{thm}\label{thm:SWY_alt} Let $k,\ell,m,n, P_\ell, V^*, \Phi$ be as above. The following is true for all $x\in \R^n$ and $e \in \R^m$ with $\|\Phi x\|_\infty \leq \mu <1$ and $\|e\|_\infty\leq \epsilon<1-\mu$ : \newline Suppose $q$ is any vector which satisfies the relation \( \Phi x + e - D^{r} u = q  \)  and \( \| u \|_{\infty} \le \gamma(r) < \infty \). Suppose further that $\hat{x}\in\R^n$ is feasible to \eqref{eq:decoder_intro} and satisfies $\|\hat{x}\|_1 \leq \|x\|_1$. Then

\begin{equation}\label{eq:l2err_intro}
\|\hat{x}-x\|_2 \leq C\left(\big(\frac{m}{\ell}\big)^{-r+1/2}\beta+\frac{\sigma_k(x)}{\sqrt k}+\sqrt{\frac{m}{\ell}}\epsilon \right), 
\end{equation}
where $C$ does not depend on $m, \ell,n$. 
\end{thm}

\subsection{Preliminaries on low-rank recovery}
A key idea in our proof is relating low-rank matrix recovery to sparse vector recovery, as was first done in \cite{oymak2011simplified}, where the following useful lemmas were presented.  

\begin{lem}[\cite{oymak2011simplified}]\label{lem:MatRIP_VecRIP}
If $\M$ satisfies the matrix-RIP of order $k$ and constant $\delta_k$, then for all unitary $U$, $V$, $\M_{U,V}$ satisfies the vector-RIP of order $k$ and constant $\delta_k$.
\end{lem}

\begin{lem}[\cite{oymak2011simplified}]\label{lem2} 
Suppose $W\in \R^{n_1\times n_2}$ admits the singular value decomposition $U_W\Sigma_W V_W^*$, and suppose $X_0 \in \R^{n_1 \times n_2}$ admits the singular value decomposition $U_{X_0}\Sigma_{X_0} {V_{X_0}}^*$. Suppose that  $\|X_0 + W\|_* \leq \|X_0\|_*$ and assume without loss of generality that $n_1 \leq n_2$. Then, there exists $X_1 = U_W \diag(z)V_W^*$ for some $z\in \R^{n_1}$ such that $\|X_1 + W\|_* \leq \|X_1\|_*.$ In particular, the choice $X_1 := -U_W \Sigma_{X_0}V_W^*$ yields the inequality.
\end{lem}

\subsection{Preliminaries on Probabalistic Tools}
Many of the classical compressed sensing results involve sampling a sparse signal with a Gaussian linear operator. It has been noticed, however, that only a handful of the special features of the Gaussian distribution are needed for these results to hold. Examples of such features include super-exponential tail decay, the existence of a moment generating function, and moments which grow ``slowly'', see \cite{VershyninNotes, FR13} for example. A class of distributions which enjoy these features is the sub-Gaussian class, which we define below.

\begin{defn}\label{def: subGauss}
Let \(X\) be a real-valued random variable. We say \(X\) is a sub-Gaussian random variable with parameter \(K\) if for all \( t \ge 0 \)
\begin{align*}
   \mathbb{P}[|X| > t] \le \exp(1-t^2/K).
\end{align*}
We say that a linear operator \( \mathcal{M} \) is sub-Gaussian if its associated matrix \(A_{\mathcal{M}}\) has entries drawn independently and identically from a sub-Gaussian distribution.
\end{defn}

The tail decay property in Definition \ref{def: subGauss} is equivalent to the \(j\)-th root of the \(j\)-th moment of a sub-Gaussian random variable \( X \) growing like \(\sqrt{j} \), or when \( \mathbb{E} X = 0 \) equivalent to the moment generating function existing over all of \( \R \). See \cite{VershyninNotes, FR13} for the details.

In the course of proving our main result, we will need to show a certain sub-Gaussian linear operator satisfies the matrix-RIP. Our proof of such will require a technique known as chaining. Talagrand makes the following definition in \cite{Talagrand}.
\begin{defn}
Given a metric space \((T, d)\), an admissible sequence of \(T \) is a collection of subsets of \(T\), \( \{T_s: s \ge 0 \} \), such that for all \( s \ge 0\), \( |T_s| \le 2^{2^s} \), and \( |T_0| = 1 \). The \( \gamma_2 \) functional is defined by
\begin{align*}
   \gamma_2(T,d) = \inf \sup_{t \in T} \sum_{s=0}^{\infty} 2^{s/2} d(t,T_s)
\end{align*}
where the infimum is taken with respect to all admissible sequences of \( T \).
\end{defn}

It is common, given the unwieldy definition above, to control the \( \gamma_2 \) functional with the well-known Dudley integral \cite{Dudley67}. In our case, we will consider a set of matrices \( \mathcal{S} \subset \R^{m\times n_1 n_2} \) equipped with the operator norm \( \| A \|_{2\to 2} = \sup_{\|x\|_2 = 1} \|Ax\|_2\). With this, we have for some universal constant \( c > 0 \)

\begin{align*}
   \gamma_2(\mathcal{S}, \| \cdot \|_{2 \to 2}) \le c \int_{0}^{d_{2 \to 2}(\mathcal{S})} \sqrt{\log(N(\mathcal{S}, \| \cdot \|_{2 \to 2}; u))} \, du
\end{align*}
where \( d_{2 \to 2}(\mathcal{S}) = \sup_{A \in \mathcal{S}} \|A\|_{2\to2}\). is the operator norm radius of the set \( \mathcal{S} \) and \( N(\mathcal{S}, \| \cdot \|_{2 \to 2}; u) \) is the covering number of \( \mathcal{S} \) with radius \( u \). 

The following useful lemma from \cite{KMRsuprema} will allow us to easily control the matrix-RIP of the linear operators $P_\ell V^*\mathcal{M}$ where $\mathcal{M}$ is sub-Gaussian. 
\begin{lem}[\cite{KMRsuprema}]
\label{lem:KMR}
Let \( \mathcal{S} \) be a set of matrices, and let \( \xi \) be a sub-Gaussian random vector with independent and identically distributed (i.i.d.) mean zero, unit variance entries with parameter $K$. Set 
\begin{align*}
   \mu &= \gamma_2(\mathcal{S}, \| \cdot \|_{2 \to 2}) \Big( \gamma_2(\mathcal{S}, \| \cdot \|_{2 \to 2}) + d_F(\mathcal{S}) \Big) \\
   \nu_1 &= d_{2 \to 2}(\mathcal{S}) \Big( \gamma_2(\mathcal{S}, \| \cdot \|_{2 \to 2}) + d_F(\mathcal{S}) \Big) \\
   \nu_2 &= d_{2 \to 2}^2(\mathcal{S}).
\end{align*}
where \( d_F(\mathcal{S}) = \sup_{A \in \mathcal{S}}\|A\|_F\) is the radius of the set \(\mathcal{S}\) with respect to the Frobenius norm. Then for all \( t > 0 \),
\begin{align*}
  \mathbb{P}\left[ \sup_{A \in \mathcal{S}} \Big| \|A \xi\|_2^2 - \mathbb{E}\|A \xi\|_2^2\Big| \ge c_1 \mu + t \right] \le 2 \exp \left( - c_2 \min \left\{ \frac{t^2}{\nu_1^2}, \frac{t}{\nu_2} \right\} \right)
\end{align*}
where the constants \( c_1, c_2 > 0 \) depend only on \( K \).
\end{lem}

\begin{lem}
\label{lem:sG concentrate}
Let \( \M: \R^{n_1 \times n_2} \to \R^m \) be a mean zero, unit variance sub-Gaussian linear map with parameter $K$, \( P_{\ell}: \R^{m} \to \R^{\ell} \) the projection map onto the first \( \ell \) coordinates, and \( V^* \in \R^{m\times m}\) a unitary matrix. Then there exist constants $C_1, C_2$ which may depend on $K$, such that for \(\ell \ge C_1 \frac{k(n_1+n_2 + 1)}{\delta_k^2} \), the operator \(\frac{1}{\sqrt{\ell}}P_{\ell} V^* \M\) has the matrix-RIP with constant \( \delta_k \) with probability exceeding \(1 - 2e^{-C_2 \ell } \).
\end{lem}  
\begin{proof}
The proof will be an application of Lemma \ref{lem:KMR}. To that end,  observe that
   \begin{align*}
      \frac{1}{\sqrt{\ell}}P_{\ell}V^* \M (X) =
      \frac{1}{\sqrt{\ell}}P_{\ell}V^*\text{diag}(\vec{X}^T) \xi
   \end{align*}
   where \(\text{diag}(\vec{X}^T) = I_{m\times m} \otimes (\vec{X})^T \), and \( \xi \in \R^{m n_1  n_2}\) is a sub-Gaussian random vector \footnote{Here, \( \otimes\) refers to the Kronecker product of matrices.}. Without loss of generality, we may assume that \( \|X\|_F = 1\) by rescaling, if necessary. It behooves us then to consider
   \begin{align*}
      \mathcal{S} = \left\lbrace \frac{1}{\sqrt{\ell}} P_{\ell} V^* \text{diag}(\vec{X}^T) : X \in \R^{n_1 \times n_2},\, \|X\|_F = 1, \text{rank}(X) \le k \right\rbrace.
   \end{align*}
    Let \( V_{i,\cdot} \) denote the \(i\)-th row of \(V\). By direct calculation, we see that
   \begin{align*}
      d_F^2(\mathcal{S}) = \sup_{X} \frac{1}{\ell} \|P_{\ell} V^* \text{diag}(\vec{X}^T) \|_F^2 
      &=  \sup_{X} \frac{1}{\ell} \|X\|_F^2 \sum_{i=1}^{m} \sum_{j=1}^{\ell} |V_{i,j}|^2 \le 1.
   \end{align*}
   Likewise, by direct calculation,
   \begin{align*}
     d_{2\to 2}^2(\mathcal{S}) &= \sup_{X} \frac{1}{\ell} \|P_{\ell} V^* \text{diag}(\vec{X}^T) \|_{2\to 2}^2 = \sup_{X} \sup_{\|w\|_2 = 1} \frac{1}{\ell} \|P_{\ell}V^* \text{diag}(\vec{X}^T)w \|_2^2 \\
&\le \frac{1}{\ell} \sup_{X} \|X\|_F^2 = \frac{1}{\ell}.
   \end{align*}
Above, the last inequality follows from the Cauchy-Schwarz inequality, or alternatively from the fact that the largest singular value of \( I \otimes \vec{X}^T \) is just the singular value of the vector $\vec{X}^T$, namely its Frobenius norm. Lemma 3.1 in \cite{candes2011tight} tells us the covering number \(N(\mathcal{S}, \| \cdot \|_F; \ep) \le \left( \frac{9}{\ep} \right)^{(n_1+n_2 + 1)k} \). Invoking Dudley's inequality and H{\"o}lder's inequality, and letting \(\mathds{1}_{B}\) denote the indicator function on the set \(B\), that is, \(\mathds{1}_{B}(x) = 1 \) for \(x \in B\) and \(0\) otherwise, we get
   \begin{align*}
      \gamma_2(\mathcal{S}, \|\cdot\|_{2\to2})& \le c_1 \sqrt{k(n_1 + n_2 + 1)} \int_{0}^{\frac{1}{\sqrt{\ell}}} \sqrt{\log\left(\frac{9}{\ep}\right)} \, d\ep \\
      &\le c_1 \sqrt{k(n_1 + n_2 + 1)} \left\| \mathds{1}_{[0,\frac{1}{\sqrt{\ell}}]} \right\|_{L^2([0,1])} \left\| \sqrt{\log\left(\frac{9}{\ep}\right)} \right\|_{L^2([0,1])}\\
      &= c_1 \sqrt{\frac{k(n_1 + n_2 + 1)}{\ell}}(\log(9) + 1)^{1/2} := c_2 \sqrt{\frac{k(n_1 + n_2 + 1)}{\ell}}.
   \end{align*}
   Putting it all together, we have 
      \begin{align*}
      \mu &\leq c_2^2 \frac{k(n_1 + n_2 + 1)}{\ell} + c_2 \sqrt{\frac{k(n_1 + n_2 + 1)}{\ell}} \le c_3\sqrt{\frac{k(n_1 + n_2 + 1)}{\ell}}  \\
      \nu_1 &\leq \frac{1}{\sqrt{\ell}} + c_2\frac{\sqrt{k(n_1+n_2+1)}}{\ell} \le c_4 \frac{1}{\sqrt{\ell}} \\
      \nu_2 &= \frac{1}{\ell}.
   \end{align*}
So invoking Lemma \ref{lem:KMR} yields for all \( t > 0 \), 
   \begin{align}\label{eq: PVM}
   \P\left[\sup_{X} \Big| \frac{1}{\ell} \|P_{\ell}V^* \M(X)\|_F^2 - \mathbb{E}\frac{1}{\ell} \|P_{\ell}V^* \M(X)\|_F^2 \Big| \ge c_5 \mu + t\right] \le 2\exp\left(-c_6 \min\left\lbrace \frac{t^2}{\nu_1^2}, \frac{t}{\nu_2} \right\rbrace \right), 
   \end{align}
   where the supremum is taken over all  $X \in \R^{n_1 \times n_2},\, \|X\|_F = 1, \text{rank}(X) \le k$.
   Note that by independence of the \( A_j \), we have
   \begin{align*}
   \mathbb{E}\frac{1}{\ell} \|P_{\ell}V^* \M(X)\|_F^2 &= \frac{1}{\ell} \mathbb{E}\sum_{i=1}^{\ell} \left( \sum_{j=1}^{m} V_{i,j} \langle A_j, X \rangle \right)^2 
   \\ &= \frac{1}{\ell} \sum_{i=1}^{\ell} \sum_{j=1}^{m} V_{i,j}^2 \mathbb{E}\langle A_j, X \rangle^2
   \\ & = \frac{1}{\ell} \|X\|_F^2 \sum_{i=1}^{\ell} \sum_{j=1}^{m} V_{i,j}^2 = \|X\|_F^2.
   \end{align*}
   Equation \eqref{eq: PVM} now becomes
   \begin{align}\label{eq: PVM reduce}
   \P\left[\sup_{X} \Big| \frac{1}{\ell} \|P_{\ell}V^* \M(X)\|_F^2 - \|X\|_F^2 \Big| \ge c_5 \sqrt{\frac{k(n_1 + n_2 + 1)}{\ell}} + t\right] \le 2\exp\left(-c_6 \min\left\lbrace c_4^{-2} t^2 \ell, t \ell \right\rbrace \right). 
   \end{align}
   Choosing \( t = \delta_k/2 \) and recalling that \( \ell \ge C_1 \frac{k(n_1+n_2 + 1)}{\delta_k^2} \) with $C_1:= 4c_5^2$,   
   equation \eqref{eq: PVM reduce} reduces to
   \begin{align*}
   \P\left[\sup_{X} \Big| \frac{1}{\ell} \|P_{\ell}V^* \M(X)\|_F^2 - \|X\|_F^2 \Big| \ge \delta_k\right] \le 2\exp\left(-C_2 \ell \right),
   \end{align*}
where $C_2:= c_6 \min\{ c_4^{-2}\delta_k^2/4,\delta_k/2\}$.
   Therefore, with probability \(1 - 2\exp\left(-C_2 \ell \right) \), we have that
   \begin{align*}
   \left| \frac{1}{\ell} \|P_{\ell}V^* \M(X)\|_F^2 - \|X\|_F^2 \right| \le \delta_k = \delta_k \|X\|_F^2
   \end{align*}
   for all $X \in \R^{n_1 \times n_2},\, \|X\|_F = 1, \text{rank}(X) \le k$. In other words, \( \frac{1}{\ell} P_{\ell} V^* \M(X) \) has the matrix RIP with constant \( \delta_k\) with high probability.
\end{proof}

\section{Recovery error guarantees}
Herein, we present our main result on the recovery error guarantees for $\SD$-quantized sub-Gaussian measurements of approximately low-rank matrices. Specifically, our results pertain to reconstruction via the constrained nuclear-norm minimization 
\begin{align}\label{eq:decoder} 
(\hat{X},\hat{\nu}) :=  \arg\min\limits_{(Z,\nu)}\|Z\|_* \  \text{ subject to } &  \|D^{-r}(\M(Z) +\nu-q )\|_2 \leq \gamma(r)\sqrt{m}  \notag \\
\text{\ \ and\ \ } & \|\nu\|_2\leq \epsilon \sqrt m
\end{align} 
where $\gamma(r)$ is the stability constant associated with the quantizer.
As such, Theorem \ref{thm: main} is a generalization of Theorem \ref{thm:SWY} to the low-rank matrix case. 
\begin{thm}[Error guarantees for stable $\SD$ quantizers]
\label{thm: main}
 Let $k,\ell,$ and $r$ be integers and let $\M:\R^{n_1\times n_2} \to \R^m$ be a mean zero, unit variance sub-Gaussian linear operator with parameter $K$. Suppose that $m\geq \ell \geq c_1 k \max\{n_1,n_2\}$. Then, with  probability exceeding $1-c_2e^{-c_3 \ell}$ on the draw of $\M$, the following holds for a stable $\SD$ quantizer with stability constant $\gamma(r)$:

For all $X\in\R^{n_1\times n_2}$, the solution ${X^\sharp}$ of \eqref{eq:decoder} where $q$ is the $\sd$ quantization of $\M(X) + e$ with $\|e\|_\infty\leq \epsilon$, satisfies
\begin{equation}\label{eq:l2err_intro}
\|{X^\sharp}-X\|_F \leq C(r)\left(\left(\frac{m}{\ell}\right)^{-r+1/2}\beta+\frac{\sigma_k(X)_*}{\sqrt k}+\sqrt{\frac{m}{\ell}}\epsilon \right).
\end{equation}
The constants $c_1, c_2, c_3, C$  do not depend on the dimensions, but may depend on $K$ and $r$.
\end{thm}
\begin{proof}
Recall that by the $\sd$ state equations, we have \[\|u\|_\infty = \| D^{-r} (\M(X) + e -q) \|_\infty \leq \gamma(r) \beta.\]
Consequently, by feasibility and optimality of $(X^\sharp, \nu^\sharp)$ respectively, we have
$ \| D^{-r} (\M(X^\sharp)+v^\sharp-q) \|_2 \leq \gamma(r) \beta \sqrt{m}$ and $\|X^\sharp\|_* \leq \|X\|_*$. 

Define $W:=X^\sharp -X$ and let $U_W \Sigma_W V^*_W$ be the singular value decomposition of $W$. Then, denoting by $U_X \Sigma_X V_X^*$ the singular value decomposition of $X$, we have by Lemma \ref{lem2}, with $X_1 = -U_W \Sigma_X V_W^*$ that  
\[ \|X_1 + W\|_* \leq  \|X_1\|_*.\]
Moreover, defining 
\[ y_1 := D^{-r}( \M(X_1) + e) + u,\]
we have by the linearity of $\M$
\begin{align}
\|  D^{-r} (\M (X_1 +W ) +v^\sharp) - y_1\|_2  &= \|  D^{-r} (\M (X_1 +W )+v^\sharp) - (D^{-r}( \M(X_1) + e) + u) \|_2 \nonumber  \\
  &= \|  D^{-r} (\M (W ) + v^\sharp-e) -  u \|_2  \nonumber \\
&= \|  D^{-r} (\M (X +W ) +v^\sharp)- (D^{-r} (\M(X) + e) + u) \|_2 \nonumber\\
&= \|  D^{-r} (\M (X^\sharp ) + v^\sharp - q) \|_2 \nonumber \\
&\leq \gamma(r)\beta \sqrt{m}. \label{eq:something}
\end{align}
Now, note that with $x_1$ denoting the vector composed of the diagonal entries of $-\Sigma_X$, we have 
\begin{align}
y_1 & =D^{-r}(\M (U_W \diag(x_1) V_W^*) + e) + u\\
&=D^{-r}(\M_{U_W,V_W} x_1 + e) + u\\
&=(D^{-r}\M)_{U_W,V_W} x_1 + D^{-r}e + u. \label{eq:y_1_2}
\end{align}
Above, we defined $(D^{-r}\M)(X):= \sum\limits_{i=1}^m \langle X, A_i \rangle D^{-r}e_i$.
Denoting by $w$ the vector composed of the diagonal entries of $\Sigma_W$, \eqref{eq:y_1_2} and \eqref{eq:something} respectively yield the inequalities \begin{equation}\|(D^{-r}\M)_{U_W,V_W} x_1 +D^{-r}e- y_1 \|_2 \leq \gamma(r)\beta \sqrt{m}.\end{equation} 
and
\begin{equation}\|(D^{-r}\M)_{U_W,V_W} (x_1+w) +D^{-r}v^\sharp - y_1 \|_2 \leq \gamma(r)\beta \sqrt{m}.\end{equation} 
Additionally, we have
that \[ \|x_1+w\|_1= \|X_1+W\|_* = \|X^\sharp\|_*  \leq  \|X\|_* = \|x_1\|_1. \] 

Thus, we have shown that the vector $x^\sharp:=x_1+w$ has a smaller  $\ell_1$ norm than $x$, and that it is feasible to \eqref{eq:decoder_intro} with $\M_{U_W,V_W}$ in place of $\Phi$ and $y_1$ in place of $D^{-r}q$. So, we are \emph{almost} ready to apply Theorem \ref{thm:SWY_alt} to $\M_{U_W,V_W}$ and conclude that 
\begin{equation} \|X^{\sharp}-X\|_F= \|W\|_F = \|w\|_2 \le C(r) \left(  \left(\frac{m}{\ell}\right)^{-r+1/2} \beta + \frac{\sigma_k(X)_*}{\sqrt k} +\sqrt{\frac{m}{\ell}}\epsilon \right).\label{eq:result}\end{equation} 
However, to do that, we must first show that  
$\frac{1}{\sqrt{\ell}}(P_\ell V^* \M)_{(U_W,V_W)}$ has the vector-RIP of order $k$ and constant $\beta < 1/9$. This, however, follows from Lemma \ref{lem:sG concentrate} where it is established that $(\frac{1}{\sqrt{\ell}}P_\ell V^* \M)$ has the required matrix-RIP, with high probability.
By Lemma \ref{lem:MatRIP_VecRIP}, this implies that $(\frac{1}{\sqrt{\ell}}P_\ell V^* \M)_{U_W,V_W} $ has the vector-RIP of order $k$ for all unitary pairs $(U_W,V_W)$, so now we may apply Theorem \ref{thm:SWY_alt} to obtain~\eqref{eq:result} and conclude the proof.

 \end{proof}
By finding the optimal quantization order $r$ as a function of the oversampling factor, as is standard in the $\SD$ literature (e.g., \cite{G-exp}, \cite{KSW12}), root-exponential error decay can attained. Corollary \ref{cor: root exp} is a precise statement to that effect. Its proof follows the same argument as Corollary 11 in \cite{SWY15}, with only the oversampling factor \( \lambda \) changed to reflect the fact that we are dealing with matrices instead of vectors.  
Next, we show that the component of the reconstruction error that is due to quantization can be made to decay root-exponentially as a function of the oversampling factor. 
\begin{cor}[Root-exponential quantization error decay]
\label{cor: root exp}
Let \( \M : \R^{n_1 \times n_2} \to \R^m \) be a mean zero, unit variance sub-Gaussian linear operator with parameter $K$ and \( X \in \R^{n_1 \times n_2} \) a rank \( k \) matrix with \( \| \M(X) \|_F \le 1\). Denote by \(Q^r_{\Sigma \Delta} \) the \(r^{th}\) order \(\Sigma\Delta \) quantizer with alphabet \( \mathcal{A} \) of step-size \(\beta\) and stability constant \( \gamma(r) \leq C^r r^r \beta \). Then there exist  constants \(c, c_1, C_1, C_2 > 0 \) that may depend on $K$, so that when
\begin{align*}
   \lambda &:= \frac{m}{\lceil ck \max(n_1, n_2) \rceil} \\
   r &:=  \left\lfloor\frac{\lambda}{2 C_1 e} \right\rfloor^{1/2} \\
   q &:= Q_{\Sigma\Delta}^{r}(\M(X)).
\end{align*}
the solution \( \hat{X} \) to \eqref{eq:decoder} satisfies \( \| \hat{X} - X \|_F \leq C_2 \beta e^{-c_1 \sqrt{\lambda}} \).
\end{cor}

Next,  Corollary \ref{cor: encoding}  shows that by projecting the quantized measurements onto a subspace of dimension $L = C k\max(n_1,n_2) \leq C'm$, where \(C, C' > 0\) are absolute constants,  we can obtain comparable reconstruction error guarantees to those of Theorem \ref{thm: main}. In turn, this allows us to obtain a reconstruction with \emph{exponentially} decaying quantization error, or distortion, as a function of the number of bits, or rate, used. We make this observation precise in Remark  \ref{rem:exp}, thereby extending the analogous result for the vector case \cite{SWY16} to our matrix setting. We comment  that, just like for sparse vectors, this exponentially decaying rate-distortion relationship is optimal for low-rank matrices over all possible encoding and decoding schemes.  

\begin{cor}[Error guarantees with encoding]
\label{cor: encoding}
Let \( \M: \R^{n_1 \times n_2} \to \R^m \) be a mean zero, unit variance sub-Gaussian linear operator with parameter $K$. Let \(B: \R^{m} \to \R^{L} \) be a Bernoulli random matrix whose entries are \( \pm 1\). Then there exist constants \(c_1, c_2, c_3, C_1, C_2 > 0 \) that may depend on $K$ and $r$,  so that whenever \( m \geq c_1 L \geq c_2 k \max(n_1, n_2) \) the following is true with probability greater than \(1 - C_1 \exp(-c_3 \sqrt{mL}) \) on the draw of \( \M \) and \(B \):

Suppose \( X \in \R^{n_1 \times n_2} \) has rank \(k \), \( \| \M(X) \|_{\infty} \le \mu < 1 \), and \(q = Q^r_{\Sigma \Delta}(\M(X) + e) \) with \( \|e\|_{\infty} \le \ep \) for \( \ep \in [0, 1-\mu) \). Then the solution of

\begin{align}\label{eq:encoder_decoder}
(\hat{X},\hat{\nu}) :=  \arg\min\limits_{(Z,\nu)}\|Z\|_* \  \text{ subject to } &  \|BD^{-r}(\M(Z) +\nu-q )\|_2 \leq 3 m \gamma(r) \notag \\
\text{\ \ and\ \ } & \|\nu\|_2\leq \epsilon \sqrt m.
\end{align} 
 satisfies
\begin{align*}
   \| \hat{X} - X \|_F \leq C_2 \left( \left( \frac{m}{L} \right)^{-r/2 + 3/4}\beta + \sqrt{\frac{m}{L}} \ep + \frac{\sigma_k(X)_*}{\sqrt{k}} \right).
\end{align*}
\end{cor}

\begin{remark}\label{rem:exp}
Let $\alpha :=\max_{a\in\A} \|a\|_\infty$. A simple calculation shows that one needs a rate of at most \( \mathcal{R} = L r \log_2(\alpha m) \) bits to store the encoded measurements. This demonstrates that in the noise free setting, and with rank $k$ matrices, the distortion $\mathcal{D} := \|\hat{X}-X\|_F$ satisfies
\begin{align}\label{eq: exp}
\mathcal{D} \leq \left(\frac{1}{\alpha L} 2^{\frac{\mathcal{R}}{Lr}}\right)^{-r/2 + 3/4}.
\end{align}
That is, the distortion decays exponentially with respect to the rate provided \( r \ge 2 \).
\end{remark}

The proof of the above corollary follows from a combination of Theorem 12 in \cite{SWY16}, which we state below, and an argument similar to the proof of Theorem \ref{thm: main}.

\begin{thm}[\cite{SWY16}] \label{thm: encode_err}
   Let \( \Phi \) be a \( m \times n \) sub-Gaussian matrix with mean zero and unit variance entries with parameter $K$, and let \( B \) be a \( L \times m \) Bernoulli matrix with \( \pm 1 \) entries. Moreover, let \( k \in \{1,..,\min\{m,n\} \}. \)
   
   Denote by \(Q_{\Sigma \Delta}^{r}\) a stable rth-order scheme with \( r > 1\), alphabet \( \A \) and stability constant \( \gamma(r) \). There exist positive constants \( C_1, C_2, C_3, C_4\) and \(c_1\) such that whenever \( \frac{m}{C_2} \ge L \ge C_1 k \log(n/k)\) the following holds with probability greater than \(1- C_3 e^{-c_1 \sqrt{mL}} \) on the draws of \( \Phi \) and \(B \):
   
   Suppose that \(x \in \R^n \), \(e \in \R^m \) with \(\| \Phi x \|_{\infty} \le \mu < 1 \) and that \(q := Q_{\Sigma \Delta}^{r}(\Phi x + e) \) where \( \|e\|_{\infty} \le \ep \) for some \(0 \le \ep < 1 - \mu \). Then the solution \( \hat{x} \) to
\begin{align}\label{eq: encoder_opt}
(\hat{x},\hat{\nu}) :=  \arg\min\limits_{(z,\nu)}\|z\|_1 \  \text{ subject to } &  \|BD^{-r}(\Phi(z) +\nu-q )\|_2 \leq 3m\gamma(r) \notag \\
\text{\ \ and\ \ } & \|\nu\|_2\leq \epsilon \sqrt m.
\end{align} 
satisfies
\begin{align*}
   \| \hat{x} - x \|_2 \leq C_4 \left( \left( \frac{m}{L} \right)^{-r/2 + 3/4}\beta + \sqrt{\frac{m}{L}} \ep + \frac{\sigma_k(x)}{\sqrt{k}} \right).
\end{align*}

\end{thm}

\begin{remark}
As before, the requirements on \( q \) can be relaxed so that it is any vector satisfying the relation \( \Phi x + \omega + D^{r} u = q  \) with $\|\Phi x\|_\infty \leq \mu <1$, $\|\omega\|_\infty\leq \epsilon<1-\mu$, and \( \| u \|_{\infty} \le \gamma(r) < \infty \).
\end{remark}

\begin{remark}
If one restricts the scope of Theorem \ref{thm: encode_err} to strictly sparse vectors, the constraint in \eqref{eq: encoder_opt} can be relaxed to \(\|BD^{-r}(\Phi(z) +\nu-q )\|_2 \leq 3\sqrt{mL} \gamma(r) \) through a Johnson-Lindenstrauss embedding argument. See the proof of Theorem 16 in \cite{SWY16} for details.
\end{remark}

\begin{proof}[Proof of Corollary \ref{cor: encoding}]
   We know (see for example  \cite{SWY16}) that  with probability exceeding \( 1 - c_2 \exp(-c_1 \sqrt{mL}) \) that \( \|B\|_{2\to 2} \le \sqrt{L} + 2 \sqrt{m} \). For such \( B \), we have
   \begin{align*}
      \| Bu \|_2 \le \|B\|_{2 \to 2}\|u\|_{\infty} \le 3m \gamma(r).
   \end{align*}
   Define, as before, the following:
   \begin{align*}
      W &= X^{\sharp} - X = U_{W} \Sigma_{W} V_{W}^*\\
      X &= U_{X} \Sigma_{X} V_{X}^*\\
      X_1 &= -U_{W} \Sigma_{X} V_{W}^*.
   \end{align*}
   By Lemma \ref{lem2}, \( \| X_1 + W \|_* \le \|X_1\|_* \). Now, define
   \begin{align*}
      y_1 = B D^{-r}\Big( \M(X_1) + e \Big) + Bu.
   \end{align*}
   By linearity of \( \M \),
   \begin{align*}
      \left\| BD^{-r}\Big( \M(X_1 + W) + \nu^{\sharp} \Big) - y_1 \right\|_2 &= \left\|B \Big(D^{-r}\big(\M(X_1 + W) + \nu^{\sharp}\big) - D^{-r}\big(\M(X_1) + e \big) - u \Big)\right\|   _2 \\
      &= \left\| BD^{-r}\Big( \M(X^{\sharp}) + \nu^{\sharp} - q \Big) \right\|_2 \le 3m\gamma(r).   
   \end{align*}
   Letting \( x_1 \) denote the vector of diagonal  elements of \( -\Sigma_X \) and \(w \) that of \( \Sigma_{W} \), we have
   \begin{align*}
      y_1 = (BD^{-r}\M)_{U_W, V_W}(x_1) + BD^{-r}e + Bu.
   \end{align*}
   Just as in the proof of the main theorem, we remark that
   \begin{align*}
    \left \| BD^{-r}\Big( \M_{U_W, V_W}(x_1) + e \Big) - y_1 \right\|_2 = \left\| Bu \right\|_2 \le 3 m \gamma(r).
   \end{align*}
   In other words, both \( x^{\sharp} := x_1 + w \) and \( x_1 \) are feasible to \eqref{eq: encoder_opt} with \( \Phi \) set to \(M_{U_W,V_W} \) and \(q \) set to \( y_1\). Moreover, we also have \( \|x_1 + w\|_1 = \| X_1 + W \|_* \le \|X_1\|_* = \|x_1\|_1 \). The result now follows by Theorem \ref{thm: encode_err}.

\end{proof}
\section{Numerical Experiments}
Herein, we present the results of a series of numerical experiments. The goal is to illustrate the performance of the algorithms studied in this paper and to compare their empirical performance to the error bounds (up to constants) predicted by the theory.  All tests were performed in MATLAB using the CVX package. One thing worth noting is that, in the interest of numerical stability and computational efficiency, we modified the constraint in \eqref{eq:decoder} to be 
\begin{align*}
\sigma_{\ell} \|P_{\ell} V^*(\M(X) + \nu - q)\|_2 \le \gamma(r) \sqrt{m},
\end{align*}
where \( \sigma_{\ell}\) is the \(\ell\)-th singular value of \( D^{-r}\). The motivation for this is that as \( r \) increases, \( D^{r} \) quickly becomes ill-conditioned. The analysis and conclusions  of Theorem \ref{thm: main} remain unchanged with the above modification. The only additional cost is computing the singular value decomposition of \( D^{-r} \) before beginning the optimization. For a fixed value of $m$ this needs to be done only once as the result can be stored and re-used. 

To construct rank $k$ matrices, we sampled \(\alpha_1, \hdots, \alpha_k \sim \mathcal{N}(0,1)\), \(u_1, \hdots, u_k \sim \mathcal{N}(0, I_{n_1 \times n_1})\), \( v_1, \hdots, v_k \sim \mathcal{N}(0, I_{n_2 \times n_2}) \), and set \( X := \sum_{i=1}^{k} \alpha_i u_i v_i^*\). We note that under these conditions \( \mathbb{E}\|X\|_F^2 = k \cdot n_1\cdot n_2 \). The measurements we collect are via a Gaussian linear operator $\M$ whose matrix representation consists of i.i.d. standard normal entries. For each experiment, we use a \textit{fixed} draw of \(M\).


First, we illustrate the decay of the reconstruction error, measured in the Frobenius norm, as a function of the order of the \(\SD\) quantization scheme for \( r = 1, 2\), and \(3\) in the noise-less setting. Experiments were run with the following parameters: \(n_1 = n_2 = 20 \), \( \ep = 0 \), alphabet step-size \(\beta = 1/2 \), rank \( k = 5\), and \( \ell = 4 \cdot k \cdot n_1 \). We let the over-sampling factor \( \frac{m}{\ell}\) range from \(5\) to \(60\) by a step size of \(5\). The reconstruction error for a fixed over-sampling factor was averaged over 20 draws of \(X\). The results are reported in Figure \ref{fig: simul_errplot} for the three choices of $r$. As Theorem \ref{thm: main}  predicts, the reconstruction error decays polynomially in the oversampling rate, with the polynomial degree increasing with $r$.

\begin{figure}[]
\centering
\begin{minipage}[b]{0.45\textwidth}
\includegraphics[scale=0.45]{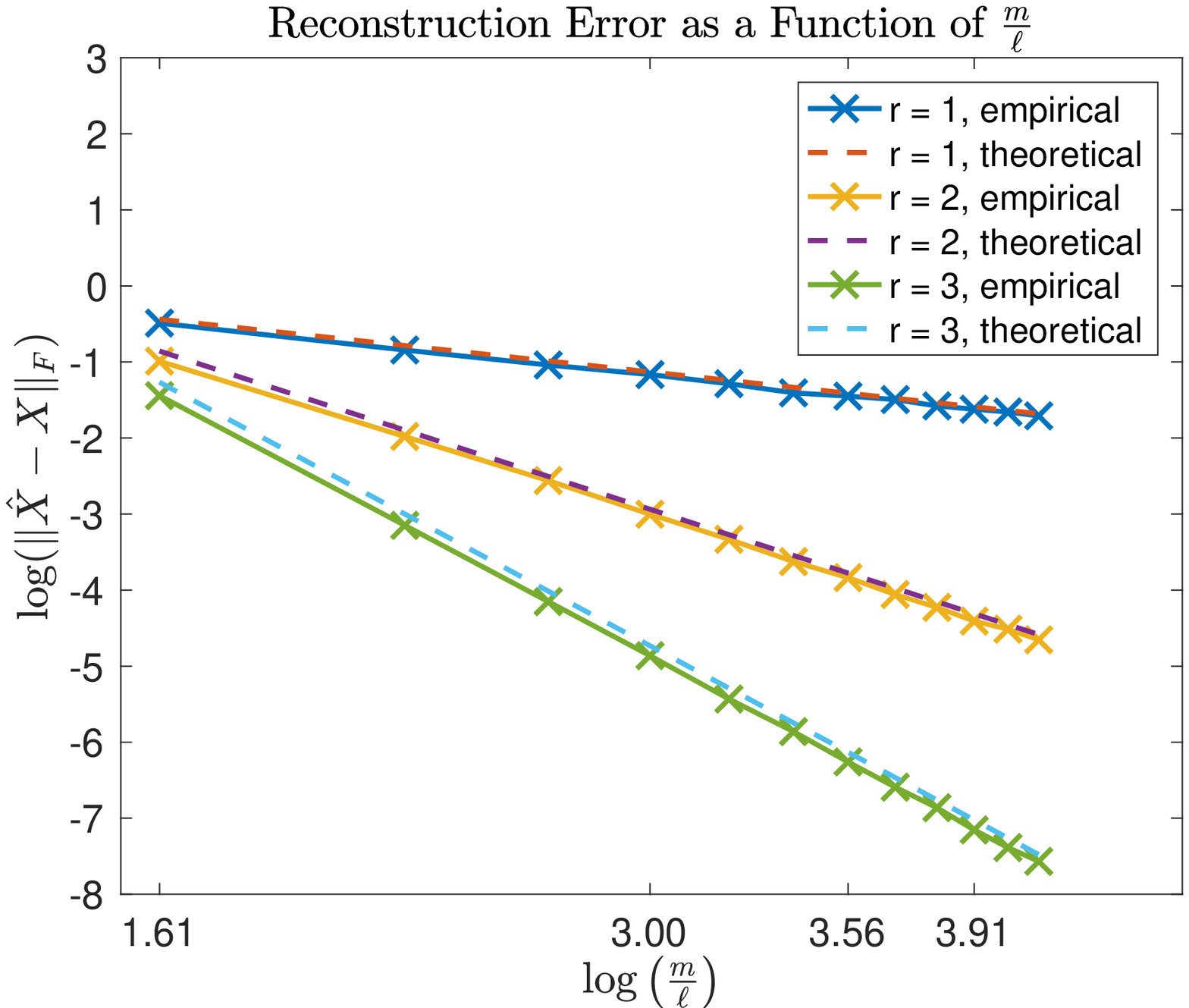}
\caption{Log-log plot of the reconstruction error as a function of the oversampling factor \( \frac{m}{\ell} \) for \( r = 1,2,3\). Polynomial relationships will appear as linear relationships.}
\label{fig: simul_errplot}
\end{minipage}
\hfill
\begin{minipage}[b]{0.45\textwidth}
\includegraphics[scale=0.45]{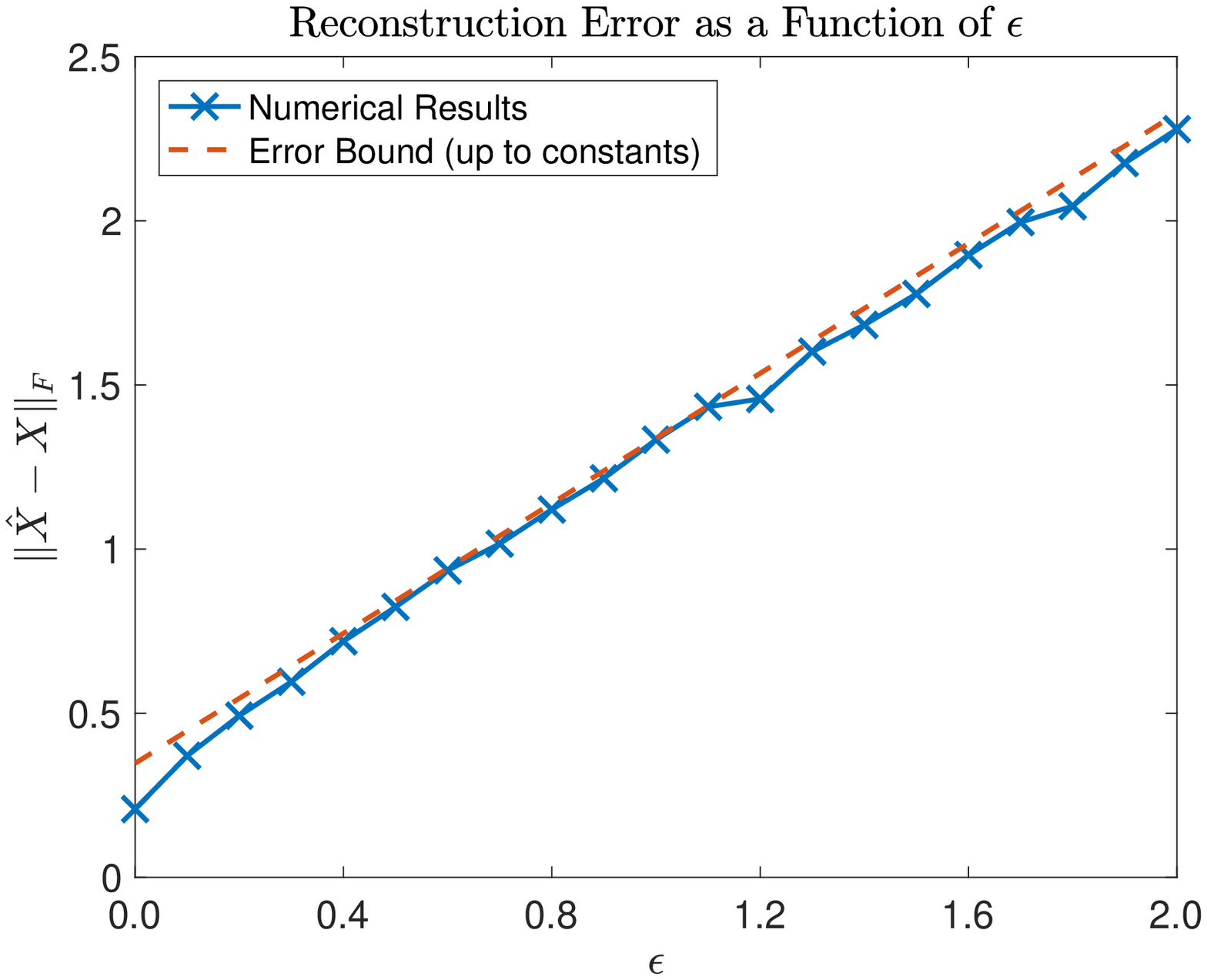}
\caption{Reconstruction error as a function of the input noise. The rank of the true matrix is 2.}
\label{fig: err_eps}
\end{minipage}
\end{figure}


To test the dependence on measurement noise, we considered reconstructing \(20 \times 20\) matrices from measurements generated by a fixed draw of \(\mathcal{M}\). 
For \(\ep \in \{0, 1/10, \hdots, 2\}\), we averaged our reconstruction error over 20 trials with noise vectors \( \nu\) drawn from the uniform distribution on \((0,1)^m\) and normalized to have \(\| \nu \|_{\infty} = \ep \).  The remaining parameters were set to the following values: \( r = 1 \), alphabet step-size \(\beta = 1/2 \), rank \(k = 2\), \( \ell = 4 \cdot k \cdot n_1\), and \( m = 2 \ell \). Figure \ref{fig: err_eps} illustrates the outcome of this experiment, which  agrees with the bound in Theorem \ref{thm: main}.

The goal of the next experiment  is to illustrate, in the context of encoding (Corollary \ref{cor: encoding}), the exponential decay of distortion as a function of the rate, or equivalently of the reconstruction error as a function of the number of bits (i.e., rate) used.  We performed numerical simulations for \( \SD \) schemes of order 2 and 3. As before, our parameters were set to the following: \( n_1 = n_2 = 20 \),  \(\beta = 1/2 \), rank of the true matrix \( k = 5 \), \( L = 4 \cdot k \cdot n_1 \), \(\ep = 0 \), and let \( \frac{m}{L}\) range from \(5\) to \(60\) by a step size of \(5\). The rate is calculated to be \( \mathcal{R} = L \cdot r \cdot \log(m)\). Again, the reconstruction error for a fixed over-sampling factor was averaged over 20 draws of \(X\). The results are shown in Figures \ref{fig: r2_encode_errplot} and \ref{fig: r3_encode_errplot}, respectively. 
The slopes of the lines (corresponding to the constant in the exponent in the rate-distortion relationship) which pass through the first and last points of each plot are \(-1.8\times 10^{-3}\) and \(-2.0\times 10^{-3}\) for \(r=2,3\), respectively. It should further be noted that the numerical distortions decay much faster than the upper bound of \eqref{eq: exp}. We suspect this to be due to the sub-optimal $r$-dependent constants in the exponent of \eqref{eq: exp}, which are likely an artifact of the proof technique in \cite{SWY16}. Indeed, there is evidence in \cite{SWY16} that the correct exponent is \(-r/2 + 1/4\) rather than \( -r/2 + 3/4\). This is more in line with our numerical exerpiments but the proof of such is beyond the scope of this paper.

\begin{figure}[]
\centering
\begin{minipage}[b]{0.45\textwidth}
\includegraphics[scale=0.45]{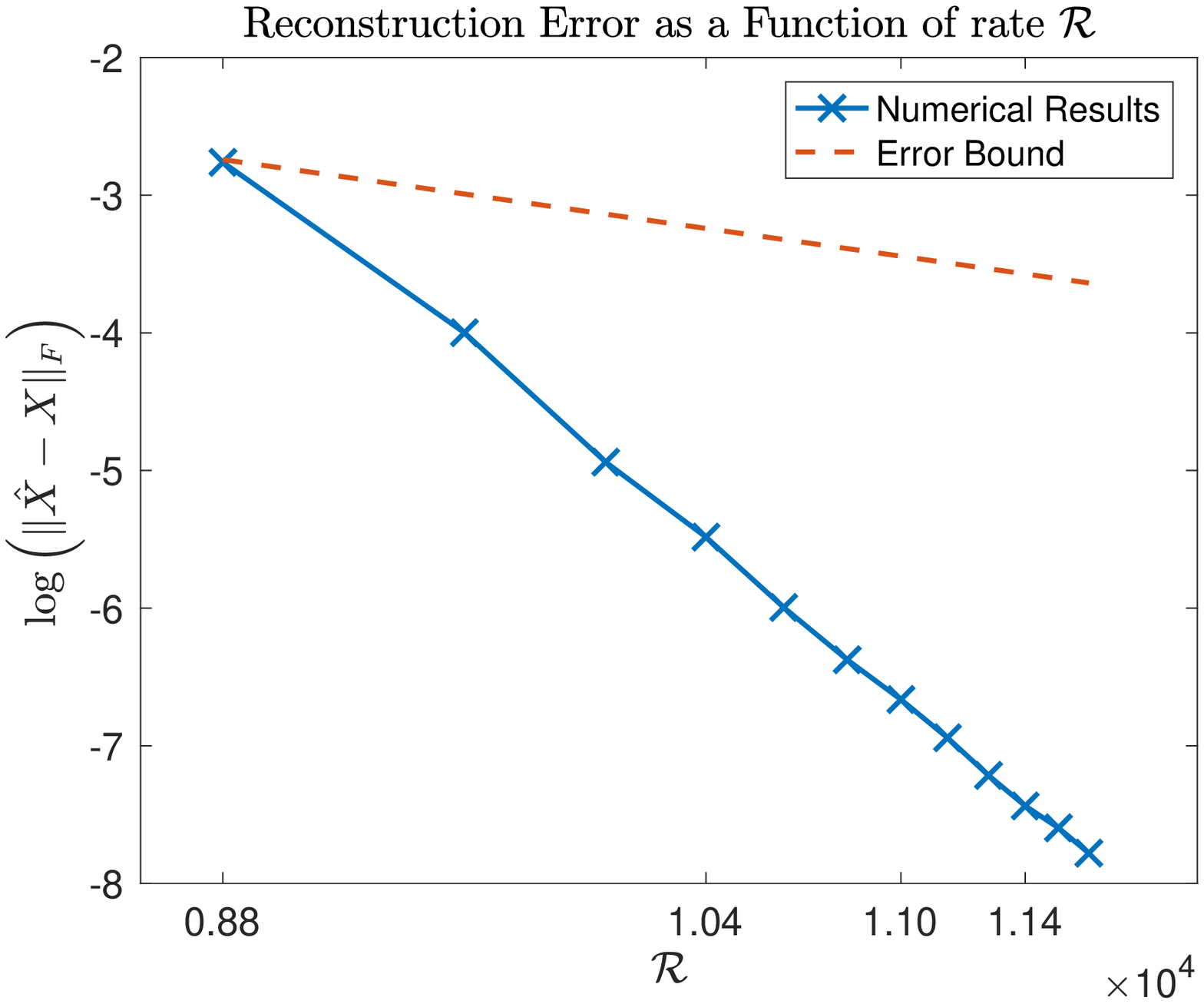}
\caption{Log plot of the reconstruction error as a function of the bit rate \( \mathcal{R} = L r \log(m) \) for \( r = 2\). Exponential relationships will appear as linear relationships.}
\label{fig: r2_encode_errplot}
\end{minipage}
\hfill
\begin{minipage}[b]{0.45\textwidth}
\includegraphics[scale=0.45]{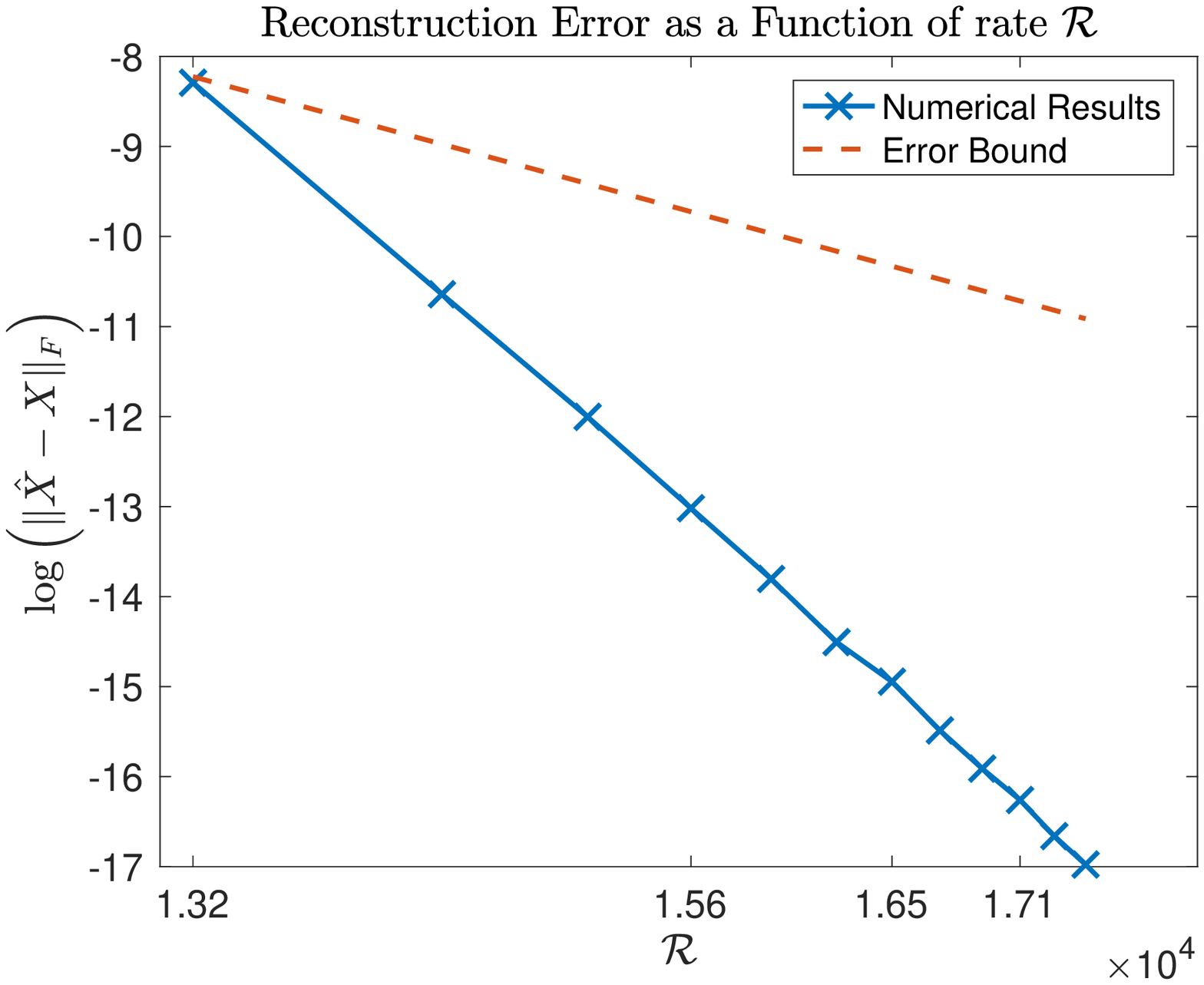}
\caption{Log plot of the reconstruction error as a function of the bit rate \( \mathcal{R} = L r \log(m) \) for \( r = 3\). Exponential relationships will appear as linear relationships.}
\label{fig: r3_encode_errplot}
\end{minipage}
\end{figure}


\section*{Acknowledgements}
\noindent RS was supported in part by the NSF via DMS-1517204. Additionally, both authors acknowledge support from a UCSD senate research grant award.

\bibliographystyle{plain}
\bibliography{lowrankSD}

\end{document}